\pdfoutput=1
\PassOptionsToPackage{unicode}{hyperref}
\PassOptionsToPackage{naturalnames}{hyperref}

\def\thm@space@setup{%
  \thm@preskip=2cm plus 1cm minus 2cm
  \thm@postskip=\thm@preskip 
}

\documentclass[11pt]{article}
\usepackage[margin=0.8in]{geometry}
\usepackage{amsthm}
\usepackage{amsfonts}
\usepackage{amssymb}
\usepackage{amsmath}
\usepackage{float}
\usepackage{caption}
\usepackage{subcaption}
\usepackage[newenum]{paralist}
\usepackage[usenames,dvipsnames]{xcolor}
\usepackage[T1]{fontenc}
\usepackage{verbatim}
\usepackage[bottom]{footmisc}
\usepackage{calligra}
\usepackage{setspace} 
\usepackage{boxedminipage}
\usepackage{multirow}
\usepackage{graphicx}
\usepackage{hyperref}
\usepackage{framed}
\usepackage{algorithm2e}
\usepackage{algpseudocode}
\newcommand{\Metivier}{$\mbox{M{\'e}tivier}$}
\newcommand{\child}{\texttt{Child}}
\newcommand{\parent}{\texttt{Parent}}
\newcommand{\INVARIANT}{\textsc{Invariant}}

\makeatletter
\usepackage{linegoal}
\algrenewcommand\algorithmicindent{2em}
\algblock[local]{}{}
\algnewcommand{\LineComment}[1]{\Statex \hskip\ALG@thistlm \parbox[t]{\linegoal}{\hangindent=1em\hangafter=1 $\triangleright$ #1}}
\makeatother
\newlength{\continueindent}
\usepackage{etoolbox}
\makeatletter
\newcommand*{\ALG@customparshape}{\parshape 2 \leftmargin \linewidth \dimexpr\ALG@tlm+\continueindent\relax \dimexpr\linewidth+\leftmargin-\ALG@tlm-\continueindent\relax}
\apptocmd{\ALG@beginblock}{\ALG@customparshape}{}{\errmessage{failed to patch}}
\makeatother

\makeatletter
 \renewcommand\section{\@startsection {section}{1}{\z@}%
 {-2.2ex \@plus -1ex \@minus -.2ex}%
 {1.2ex \@plus.1ex}%
 {\normalfont\Large\bfseries}}
\makeatother

\makeatletter
 \renewcommand\subsection{\@startsection {subsection}{1}{\z@}%
 {-2ex \@plus -1ex \@minus -.2ex}%
 {1ex \@plus.1ex}%
 {\normalfont\large\bfseries}}
\makeatother

 \makeatletter
 \newtheorem*{rep@theorem}{\rep@title}
 \newcommand{\newreptheorem}[2]{%
 \newenvironment{rep#1}[1]{%
 \def\rep@title{#2 \ref{##1}}%
 \begin{rep@theorem}}%
 {\end{rep@theorem}}}
 \makeatother

 \newtheorem{theorem}{Theorem}[section]
 \newtheorem{lemma}[theorem]{Lemma}

 \newreptheorem{theorem}{Theorem}
 \newreptheorem{lemma}{Lemma}
 \newreptheorem{corollary}{Corollary}
\hypersetup{colorlinks,linkcolor=blue,filecolor=blue,citecolor=blue, urlcolor=blue,pdfstartview=FitH}

\begin{document}
\title{Using Read-$k$ Inequalities to Analyze a Distributed MIS Algorithm\thanks{This work is supported in part by National Science Foundation grant CCF-1318166.}}
\author{Sriram V. Pemmaraju \hspace{4em} Talal Riaz\\ \small{Department of Computer Science, The University of Iowa, Iowa City, IA 52242}\\ 
\texttt{\{sriram-pemmaraju, talal-riaz\}@uiowa.edu}}
\date{}
\thispagestyle{empty}
\pagenumbering{gobble}
\maketitle
\begin{abstract}
Until recently, the fastest distributed MIS algorithm, even for simple graph classes
such as unoriented trees that can contain large independent sets within neighborhoods, 
has been the simple randomized algorithm discovered independently
by several researchers in the late 80s. This algorithm (commonly called
\textit{Luby's algorithm}) computes an MIS of an $n$-node graph in $O(\log n)$
communication rounds (with high probability). 
This situation changed when Lenzen and Wattenhofer (PODC 2011)
presented a distributed (randomized) MIS algorithm for unoriented trees running in
$O(\sqrt{\log n}\cdot \log\log n)$ rounds. This algorithm was slightly improved by Barenboim
et al.~(FOCS 2012), resulting in an $O(\sqrt{\log n \cdot \log\log n})$-round (randomized)
MIS algorithm for trees.
At their core, these algorithms still run Luby's algorithm, but only up to the point at which
the graph has been ``shattered'' into small connected components that can be 
independently processed in parallel.

The analyses of these tree MIS algorithms critically depends on ``near independence''
among probabilistic events, a feature that arises from the tree structure of the network.
In their paper, Lenzen and Wattenhofer express hope that their algorithm and analysis could
be extended to graphs with bounded arboricity. We show how to do this in the current paper.
By using a new tail inequality for \textit{read-k families} of random variables due
to Gavinsky et al.~(\textit{Random Struct Algorithms}, 2015), we show how to deal with dependencies
induced by the recent tree MIS algorithms when they are executed on bounded arboricity
graphs. Specifically, we analyze a version of the tree MIS algorithm of Barenboim et al.~and show
that it runs in $O(\mbox{poly}(\alpha) \cdot \sqrt{\log n \cdot \log\log n})$ rounds in the 
$\mathcal{CONGEST}$ model for graphs with arboricity $\alpha$.

While the main thrust of this paper is the new probabilistic analysis via read-$k$ inequalities, 
we point out that
at for small values of $\alpha$, this algorithm is faster than the MIS algorithm of Barenboim
et al.~specifically designed for bounded arboricity graphs.
In this context, it should be noted that recently (in SODA 2016) Gaffari presented a novel distributed 
MIS algorithm for general graphs that runs in $O(\log \Delta) + 2^{O(\sqrt{\log\log n})}$ rounds
and a corollary of this algorithm is an $O(\log \alpha + \sqrt{\log n})$-round MIS algorithm
on graphs with arboricity $\alpha$.



\end{abstract}
\clearpage
\pagenumbering{arabic}

\section{Introduction}
A set of nodes in a graph is said to be \textit{independent} if no
two nodes in the set are adjacent.
A \textit{maximal independent set (MIS)} is an independent set that
is maximal with respect to inclusion.
Computing an MIS is a fundamental problem in
distributed computing because it nicely captures the essential
challenge of symmetry breaking and also for its myriad applications to
other problems.
The fastest algorithm for MIS is a simple, randomized algorithm
discovered more than 25 years ago, independently by several
researchers \cite{AlonBabaiItai,LubySICOMP86,IsraeliItai}.
This algorithm computes an MIS for an $n$-node graph in $O(\log n)$
communication rounds with high probability (whp), i.e., with probability at least $1 - 1/n$.
The essence of this algorithm is that in each round, each still-active node
tentatively joins the MIS with some probability
and then either backs off from this choice or makes it permanent depending
on whether neighboring nodes have made conflicting choices.
Following popular usage, we refer to this as \textit{Luby's algorithm}.
More recently, \Metivier\ et al.~\cite{Metivieretal} proposed a variant of
Luby's algorithm in which in each round, each still-active node
$v$ picks a \textit{priority}, a real number $r(v)$
uniformly at random from $[0, 1]$ and joins the MIS if $r(v)$ is greater
than the priorities chosen by all neighbors.
This algorithm also runs in $O(\log n)$ rounds whp~\cite{Metivieretal}\footnote{In fact,
Algorithm A in Luby's 1986 paper \cite{LubySICOMP86} is essentially identical to the algorithm of 
\Metivier\ et al., the only difference being that in Luby's Algorithm, vertices choose priorities
from the range $\{1, 2, \ldots, n^4\}$. What we refer to as Luby's algorithm above appears
as Algorithm B in Luby's paper.}.

In PODC 2011, Lenzen and Wattenhofer \cite{LenzenWattenhofer} showed that
an MIS in an $n$-node \textit{unoriented} tree can be computed in
$O(\sqrt{\log n} \cdot \log\log n)$ rounds whp.
Note that if a tree is consistently oriented (i.e., the tree is rooted at
an arbitrary node and all nodes know their parent with respect to this
root) then an MIS can be computed in $O(\log^* n)$ rounds
using the deterministic coin tossing technique of Cole and Vishkin \cite{ColeVishkin}.
The first phase of the Lenzen-Wattenhofer algorithm is just the algorithm
of \Metivier\ et al.~and in a sense all the important hard work happens in this phase.
The running time analysis of the algorithm is sophisticated and depends critically on
the fact that the tree structure ensures that there are very few dependencies among probabilistic events in the algorithm.
There have been previous sublogarithmic-round MIS algorithms for special
graph classes (e.g., the $O(\log^* n)$-round MIS algorithm on growth-bounded graphs \cite{WattenhoferSchneiderPODC2008}),
but not for graphs that can have arbitrarily large independent sets in neighborhoods. 
Thus, in a sense, the Lenzen-Wattenhofer MIS result is a breakthrough
because it shows that MIS can be computed in sublogarithmic rounds
even in settings where neighborhoods can have arbitrarily many independent nodes.
More recently in FOCS 2012, Barenboim et al.~\cite{BEPS12FOCS,BEPS12arxiv} presented
a tree MIS algorithm (similar to the Lenzen-Wattenhofer algorithm) that runs in 
$O(\sqrt{\log n \cdot \log\log n})$ rounds whp,
improving the running time of the Lenzen-Wattenhofer algorithm slightly.
This tree MIS algorithm also uses the algorithm of \Metivier\ et al.~to do a significant
portion of the work.

A natural question that arises from the analyses of these tree MIS algorithms is
whether the algorithms and analyses can be extend to
\textit{bounded arboricity graphs}. Lenzen and Wattenhofer raise this question
at the end of the ``Introduction'' section in their paper \cite{LenzenWattenhofer}.
A graph $G$ is said to have \textit{arboricity} $\alpha$ if $\alpha$ is the minimum number of forests that
the edges of $G$ can be partitioned into.
From this it follows that the edges of a graph with arboricity $\alpha$ can be
oriented in such a manner that each node has at most $\alpha$ outgoing edges.
Clearly, forests have arboricity 1, but the
family of graphs with constant arboricity is quite rich and
includes all planar graphs, graphs with constant treewidth,
graphs with constant genus, family of graphs that exclude a
fixed minor, etc.
Unfortunately, the Lenzen-Wattenhofer analysis and the Barenboim et al.~analysis 
runs into trouble for graphs with even \textit{constant arboricity} because of the nature of dependencies between
probabilistic events in the algorithm.
The issue is common to both algorithms because it arises in the portions of the algorithms
that rely on the algorithm of \Metivier\ et al.

The source of the difficulty can be explained as follows.
Even though these algorithms run on unoriented trees,
\textit{for the purposes of analysis} it can be assumed that the input tree is rooted at an
arbitrary node.
Because the graph is a tree, probabilistic events at children of a node $v$ are
essentially independent, the only slight dependency being
caused by the interaction via their parent, namely $v$.
For graphs with arboricity greater than 1 the dependency
structure among the probabilistic events can be much more complicated.
Suppose (for the purposes of the analysis) that we orient the edges of an arboricity-$\alpha$ graph such that
each node has at most $\alpha$ out-neighbors. Let us call the out-neighbors of
a node $v$ its \textit{parents} (denoted $\parent(v)$) and the in-neighbors, its \textit{children}
(denoted $\child(v)$).
For a node $v$, consider the set $\child(v)$ and the dependencies among
probabilistic events at nodes in $\child(v)$.
The events we are referring to are of the type ``$w$ joins the MIS'' or ``a neighbor of $w$ joins the MIS'' for $w \in \child(v)$.
Even though each node has at most $\alpha$ parents, a node $w \in \child(v)$ may share children with every 
other node in $\child(v)$ and as a result there could be dependencies between events at $w$ and events at any
of the other nodes in $\child(v)$.
Thus it is not clear how to take advantage of the structure of bounded arboricity
graphs in order to mimic the analysis in \cite{LenzenWattenhofer,BEPS12FOCS,BEPS12arxiv}.

The main purpose of this paper is to show that recent results on
\textit{read-$k$ families of random variables} deal with roughly this type of
dependency structure and therefore provide a new approach to 
analyzing algorithms in the style of \Metivier\ et al.~with more complicated dependency structure.
Using analysis based on \textit{read-$k$ inequalities} (see the next section), we show that the tree MIS
algorithms of Lenzen and Wattenhofer \cite{LenzenWattenhofer} and Barenboim et al.~\cite{BEPS12arxiv,BEPS12FOCS} 
work for bounded arboricity graphs as well.
We believe that this analytical tool may be new to the distributed computing community, but 
will prove useful for the analysis randomized distributed algorithms in general.

\subsection{Read-\textit{k} Inequalities}
We now define a \textit{read-$k$ family} of random variables.
Let $\{Y_j \mid 1 \le j \le n\}$ be a set of random variables such that each
random variable $Y_j$ is a function of some subset of the set of independent
random variables $\{X_i \mid 1 \le i \le m\}$.
For each $1 \le j \le n$, let $P_j \subseteq \{1, 2, \ldots, m\}$, let
$f_j$ be a boolean function of $\{X_i \mid i \in P_j\}$, and define
$Y_j := f_j((X_i)_{i \in P_j})$.
The collection of random variables $Y_j$ is called a \textit{read-$k$ family}
if every $1 \le i \le m$ appears in at most $k$ of the $P_j$'s.
In other words, each $X_i$ is allowed to influence at most $k$ of the $Y_j$'s.
Note that the $Y_j$'s can have a complicated dependency structure amongst themselves --
it is their dependency on the $X_i$'s that is bounded.
For example, the dependency graph of the $Y_j$'s can even be a clique!


We are now ready to state the first of the two read-$k$ inequalities 
from Gavinsky et al.~\cite{GLSS15} that we use.
This inequality provides a bound on the conjunction of a collection of events whose indicator
variables form a read-$k$ family.
\begin{theorem}[Theorem 1.2,~\cite{GLSS15}]
\label{theorem:conjunctionReadKFamily}
Let $Y_1, Y_2, \cdots,Y_n$ be a family of read-$k$ indicator variables
with $\Pr[Y_i = 1] = p$. Then,
$\Pr[Y_1 = Y_2 = \cdots  = Y_n = 1] \leq p^{n/k}$.
\end{theorem}
\noindent
If the $Y_j$'s were independent, then the probability that $Y_1 = Y_2 = \cdots
= Y_n = 1$ would simply be $p^n$.
Thus Theorem \ref{theorem:conjunctionReadKFamily} is essentially saying that the read-$k$ family structure
of the dependencies among the $Y_j$'s allows us to obtain an upper bound
on the probability that is an exponential factor $1/k$ worse than
what is possible had the $Y_j$'s been independent.

Gavinsky et al.~\cite{GLSS15} use Theorem \ref{theorem:conjunctionReadKFamily} and information-theoretic
arguments to derive the following tail inequality on the sum of indicator
random variables that form a read-$k$ family.
\begin{theorem}[Theorem 1.1, ~\cite{GLSS15}]
\label{thm:readk}
Let $Y_1, \cdots, Y_n$ be a family of read-$k$ indicator variables with
$\Pr[Y_i = 1] = p_i$. Define $p:= \frac{1}{n} \sum_{i=1}^n p_i$ and $Y :=
\sum_{i=1}^n Y_i$. Then for any
$\epsilon, \delta > 0$,
\begin{align}
\Pr(Y \leq (p- \epsilon)n) \leq \exp\left({-2\epsilon^2 \frac{n}{k}}\right)\label{eqn:readk1} \\
\Pr(Y \leq (1- \delta) E[Y]) \leq \exp\left({-\frac{\delta^2 E[Y]}{2k} }\right) \label{eqn:readk2}
\end{align}
\end{theorem}

\noindent
Gavinsky et al.~only state Form (1) of the tail inequality in their paper. But, Form (2) is more
convenient for us and it is fairly routine to derive this from Form (1). See \cite{S13} for the 
derivation.
As in Theorem \ref{theorem:conjunctionReadKFamily}, these tail inequalities are also an
exponential $1/k$ factor worse than corresponding Chernoff bounds that we might have used,
had the $Y_j$'s been independent.
Gavinsky et al.~\cite{GLSS15} also point out that these tail inequalities are more general than
those that can be obtained by observing that $Y$ is a $k$-Lipschitz function and using standard
Martingale-based arguments such as Azuma's inequality.

To see that the above tools are well-suited for analyzing 
algorithms in the style of \Metivier\ et al.~on bounded arboricity graphs, let us reconsider the 
situation described earlier.
Consider a graph $G$ with arboricity $\alpha$ and
fix an arbitrary node $v$ in $G$, and consider the set $\child(v)$ of children of $v$.
For the moment, ignore edges among nodes in $\child(v)$ and also ignore the influence of parents ($v$ and other parents)
on nodes in $\child(v)$, thus focusing only on the children of nodes in $\child(v)$.
For a node $w \in \child(v)$, let $Y_w$ be an indicator variable for a probabilistic
event at node $w$.
Now suppose that $Y_w$ depends on independent random choices made by $w$
and its children.
For example, $Y_w$ could be a boolean variable indicating the event that the priority of
$w$ is larger than the priorities of children of $w$.
This models the situation in the algorithm of
\Metivier\ et al.~\cite{Metivieretal}, where $w$ joining the MIS depends
on random real values (independently) chosen by $w$ and its children. (Recall that we are
ignoring parents for the moment.)
The structure of an arboricity-$\alpha$ graph and the associated edge-orientation
ensures that each node has at most $\alpha$ parents and therefore the random
choice at each node can influence at most $\alpha$ of the $Y_w$'s.
Thus the set $\{Y_w \mid w \in \child(v)\}$ forms a read-$\alpha$ family and we can apply Theorems 
\ref{theorem:conjunctionReadKFamily} and \ref{thm:readk} to bound $Pr(\cap_w Y_w = 1)$ and to show
that $\sum_w Y_w$ is concentrated about its expectation.

The above example illustrates the simplest application of read-$k$ inequalities in our analysis.
Somewhat surprisingly, we use read-$k$ inequalities to evaluate probabilistic interactions between
a node and its parents also. 
This may seem impossible to do given that a parent can have arbitrarily many children
and thus a random choice at a parent can influence events at arbitrarily many children.
However, in our algorithm nodes with extremely high degree opt out of the competition (temporarily) and this
turns out to be sufficient to bound the number of children a parent can influence, leading
to our use of read-$k$ inequalities, with appropriate $k$, to analyze the interaction between 
nodes and their parents.
Finally, our analysis also relies on interactions between a node and its grandchildren, leading to
our use of read-$\Theta(\alpha^2)$ families as well.

\subsection{Our Result}

We apply a read-$k$-inequality-based analysis to the execution of the tree MIS algorithm of Barenboim 
et al.~\cite{BEPS12arxiv,BEPS12FOCS} on bounded arboricity graphs.
We could have chosen to analyze the tree MIS algorithm of Lenzen and Wattenhofer,
but for reasons of exposition we use the algorithm of Barenboim et al.
We present an algorithm that we call \textsc{BoundedArbIndependentSet},
which is essentially identical to the \textsc{TreeIndependentSet} algorithm
of Barenboim et al.~(Section 8, \cite{BEPS12arxiv}), except for parameter values (which now depend on the arboricity $\alpha$).
Specifically, we show the following result.

\begin{theorem}
\label{theorem:mainResult}
The tree MIS algorithm of Barenboim et al.~\cite{BEPS12arxiv,BEPS12FOCS} (with appropriate parameter values)
can be used to compute an MIS in the $\mathcal{CONGEST}$ model on the family of graphs
with arboricity $\alpha$  in $O(\mbox{poly}(\alpha) \cdot \sqrt{\log n \cdot \log\log n})$ rounds, whp.
\end{theorem}
\noindent
This result can also be seen as an improvement over the MIS result
on bounded arboricity graphs due to Barenboim et el.~\cite{BEPS12arxiv,BEPS12FOCS}.
In their paper, Barenboim et al.~have a separate algorithm (distinct from their tree MIS) algorithm
that computes an MIS on graphs with arboricity $\alpha$ in $O(\log^2 \alpha + \log^{2/3} n)$
rounds.
The dependency on $n$ of the running time of our algorithm is asymptotically better, implying that for
small $\alpha$ (i.e., $\alpha = O(\log^{c} n)$ for a small enough constant $c$) our
algorithm is asymptotically faster.
In our subsequent calculations the degree of polynomial in $\alpha$ in the running time comes out to be 9.
It is not difficult to reduce this degree, but it does seem difficult with the current algorithm 
to improve the dependency on $\alpha$ to something better than a polynomial and 
to replace the multiplication between the $\mbox{poly}(\alpha)$-term and the 
$\sqrt{\log n \cdot \log\log n}$-term by an addition.

Recently, in SODA 2016 Ghaffari has presented a novel MIS algorithm \cite{GhaffariSODA16}
that runs in $O(\log \Delta) + 2^{O(\sqrt{\log\log n})}$ rounds on any $n$-vertex graph with
maximum degree $\Delta$.
In Luby's MIS algorithm, a node's ``desire'' to join the MIS is simple function of its degree
with respect to the still-active nodes in the graph.
In Gaffari's MIS algorithm each node explicitly maintains a \textit{desire-level} that is initially
set to 1/2, but is updated in each iteration depending on the aggregate desire-level of nodes its
neighborhood.
Using techniques from \cite{BEPS12arxiv,BEPS12FOCS}, Gaffari obtains, as a corollary of this main result,
an $O(\log \alpha + \sqrt{\log n})$-round
MIS algorithm for $n$-vertex graphs with arboricity $\alpha$.
This of course dominates the round complexity our algorithm for all values of $\alpha$ and $n$.
Thus the main contribution of this paper is not the fastest distributed MIS algorithm for bounded arboricity graphs,
but it is (i) introducing the use of read-$k$ inequalities for the analysis of randomized distributed 
algorithms and (ii) showing that recent tree MIS algorithms are effective for bounded arboricity graphs
as well, but need more sophisticated analysis.

\section{MIS Algorithm for Bounded Arboricity Graphs}
\label{section:algorithm}
We start by presenting an algorithm that we call \textsc{BoundedArbIndependentSet},
which is essentially identical to the \textsc{TreeIndependentSet} algorithm
of Barenboim et al.~(Section 8, \cite{BEPS12arxiv}), except for parameter values
($\Theta$, $\Lambda$, $\rho_k$) which now depend on $\alpha$ as well.
We emphasize this point because we are essentially analyzing the \textsc{TreeIndependentSet} algorithm (via a new approach
based on the read-$k$ inequalities), but with bounded arboricity graphs as input.

\RestyleAlgo{boxruled}
\begin{algorithm}
\caption{\textsc{BoundedArbIndependentSet}(Graph $G$):}
{\small
\begin{algorithmic}[1]
\State Initialize sets $I,B \subseteq V(G)$: $I \leftarrow \phi$;  $B \leftarrow \phi$  \DontPrintSemicolon\\
\For{each scale $k$ from 1 to  $\Theta:= \left \lfloor\log \left (  \frac{\Delta}{1176\cdot 16\alpha^{10}\ln^2\Delta}    \right )\right\rfloor$ }{
\begin{itemize}
 \item[] Initialize $\rho_k \leftarrow 8\ln \Delta \cdot \Delta/2^{k+1}$
 \item[2(a)] Execute $\Lambda := \lceil p\cdot 8\alpha^2 (32\alpha^6 + 1) \cdot \ln(260 \alpha^4 \ln^2 \Delta) \rceil$ times \;
\begin{itemize}
\item Each node $v \in V_{IB}$ chooses a priority $r(v)$:
 \noindent \[ r(v) \leftarrow \begin{cases}
  0, \text{ if } \deg_{IB}(v)   > \rho_k \\ 
        \text{a real in} (0,1) \text{ chosen uniformly at random otherwise }
    \end{cases} \]
\item $I \leftarrow I \cup \{v \in V_{IB} \vert r(v) > max \{r(w) \vert w \in \Gamma_{IB} (v)\}\}$ \;
\item $V_{IB} \leftarrow V_{IB}\setminus(I \cup \Gamma_{IB} (I))$
\end{itemize}
\item[2(b)] Each node $v$ is marked ``bad'' if $|\{ w \in  \Gamma_{IB}(v) \vert deg_{IB}(w)  > \Delta/2^k + \alpha \}| > \Delta/2^{k+2}$
        $B\leftarrow B \cup \{ v \in V_{IB} \mid v \mbox{ is marked  ``bad''}\}$ \;
        $V_{IB} \leftarrow V_{IB} \setminus B $
\end{itemize}
}\State \textbf{return} $(I,B)$
\end{algorithmic}}
\end{algorithm}

The algorithm (see \textbf{Algorithm 1}) begins by initializing two sets $I$ and $B$ as empty. 
$I$ denotes the set of nodes which have joined the MIS and
$B$ will store a set of so-called ``bad'' nodes. 
As nodes join $I$ and $B$, they exit the algorithm, i.e., become \textit{inactive}. 
In addition, neighbors of nodes in $I$ also exit the algorithm and become inactive.
We use $V_{IB}$ to denote the set of nodes which are currently active.
Let $\Gamma_{IB}(u)$ represent the neighborhood of a node $u$ restricted to nodes in $V_{IB}$. 
Let $deg_{IB}(u)$ denote $|\Gamma_{IB}(u)|$.
Similarly, for any subset $S \subseteq V$ of nodes, let $\Gamma_{IB}(S)$ denote 
$\cup_{u \in S} \Gamma_{IB}(u)$.
The algorithm proceeds in $\Theta:= \left \lfloor\log \left(\frac{\Delta}{1176\cdot 16\alpha^{10}\ln^2\Delta}    \right)\right\rfloor$ \textit{scales}.
For any scale $k$, $1 \le k \le \Theta$, a node in $V_{IB}$ that has degree more than 
$\Delta/2^{k}+\alpha$ is called a \textit{high degree} node for that scale. 
In each scale, we start by performing $O(\alpha^8 (\log \alpha + \log\log \Delta))$
\textit{iterations} of the \Metivier\ et al.~MIS algorithm \cite{Metivieretal}.  
The exact number of iterations is 
$\lceil p \cdot 8\alpha^2 (32\alpha^6 + 1) \cdot \ln(260 \alpha^4 \ln^2 \Delta) \rceil$ and denoted by the parameter $\Lambda$, and where $p$ is a large enough constant whose value will be fixed later.

In a single iteration, every node $v \in V_{IB}$ chooses a real number $r(v) \in [0,1)$ called 
a \textit{priority}. 
If $v$ has more than $\rho_k := 8\ln\Delta \cdot \Delta/2^{k+1} $ neighbors in any iteration, 
its priority is (deterministically) set to $0$, otherwise, it chooses a priority uniformly at random in $(0,1)$. 
In any iteration, a node $u$ is called \textit{competitive}, if $r(u)$ is chosen randomly in that iteration. 
If in an iteration, $v$ chooses a priority greater than the priority of any node 
in its neighborhood in $V_{IB}$, it joins $I$. 
After each iteration, nodes in $I$ and neighbors of these nodes (i.e., $\Gamma_{IB}(I)$) are removed from $V_{IB}$.
If, after $\Lambda$ iterations in the current scale, 
a node $v \in V_{IB}$ has more than $\Delta/2^{k+2}$ high-degree neighbors then it is designated
a ``bad'' node and added to the set $B$
It is worth emphasizing the fact that this algorithm has no access to an edge-orientation or
a forest-decomposition of the given $\alpha$-arboricity graph.
We use the existence of an edge-orientation extensively in our analysis, 
but it plays no role in the algorithm.

\subsection{Finishing Up}

The algorithm returns an independent set $I$ (which need not be maximal),
and a set $B$ of ``bad'' nodes. Also, the set $V_{IB}$ need not be empty at the end of the algorithm
and so after Algorithm \textsc{BoundedArbIndependentSet} has completed,
we still need to process the sets $B$ and $V_{IB}$.

Our main contribution in this paper is a new analysis of \textsc{BoundedArbIndependentSet} that culminates in
Theorem \ref{theorem:lowProbability}, showing
that any node joins $B$ with probability at most $1/\Delta^{2p}$.
(Here $p$ is the constant that is used in determining $\Lambda$, the number of iterations
of \textsc{BoundedArbIndependentSet}.)
The fact that each node joins $B$ with very low probability implies (as shown by Barenboim 
et al.~\cite{BEPS12arxiv} and restated in Lemma \ref{lemma:smallConnectedComponents}) that
with high probability all connected components in the graph induced by $B$ are small.
These components induced by $B$ can be processed in parallel, with each component being processed
by a deterministic algorithm (since each component is small).

Nodes that remain in $V_{IB}$ have the property that they do not have too many high degree
neighbors. Otherwise, they would have been placed in $B$. 
Thus $V_{IB}$ can be partitioned into two sets $V_{hi}$ and $V_{low}$ such that
the graphs induced by each of these sets has small maximum degree.
Then, by using an alternate MIS algorithm that finishes quickly as a function of the maximum degree,
we process nodes in $V_{hi}$ and $V_{low}$ (one set after the other) to 
complete the MIS computation.
All these steps that ``finish up'' the algorithm run in the $\mathcal{CONGEST}$ model and
we describe these in greater detail in Section \ref{section:finishingUP}.

It is immediate that the round complexity of \textsc{BoundedArbIndependentSet} is $O(\alpha^8(\log \alpha 
+ \log\log \Delta) \cdot \log \Delta)$.
The rest of the algorithm (described informally above) takes an additional $O(\alpha^2 +  (\log\log\Delta)^2
+ \log\log n \cdot \log \alpha)$ rounds whp (see Section \ref{section:finishingUP}).
To get a round complexity bound that is exclusively in terms of $n$, we use a degree-reduction 
result of Barenboim et al.~(Theorem 7.2 \cite{BEPS12arxiv}) that runs in $O(\sqrt{\log n \cdot
\log\log n})$ rounds in the $\mathcal{CONGEST}$ model and yields a graph with maximum degree
at most $\alpha \cdot 2^{\sqrt{\log n \cdot \log\log n}}$
(see Section \ref{section:finishingUP} for details).

\begin{theorem}
Using \textsc{BoundedArbIndependentSet} we can compute an MIS on a graph with arboricity $\alpha$ in
$O\left(\alpha^8(\log \alpha + \log\log \Delta) \cdot \log \Delta + \log\log n \log \alpha\right)$ rounds whp.
This leads to an algorithm that computes an MIS on a graph with arboricity $\alpha$ in $O(\alpha^9 \sqrt{\log n \cdot \log\log n})$ rounds whp.
\end{theorem}

\section{Analysis of \textsc{BoundedArbIndSet}}
\label{section:analysis}

We start with an overview of our analysis.
At a high level, the organization of our analysis is similar to the analysis of \textsc{TreeIndependentSet}. 
The analysis is centered around showing that the following invariant is maintained (at the end of each scale) at every active node,
with sufficiently high probability.
\begin{framed}
\vspace{-1mm}
\INVARIANT: At the end of scale $k$, for all $v \in V_{IB}$, 
$$\left\vert \{ w \in  \Gamma_{IB}(v) \vert \deg_{IB}(w) > \Delta/2^k + \alpha  \}\right\vert \leq \Delta/2^{k+2}$$
\vspace{-7mm}
\end{framed}
\noindent
The \textsc{Invariant} bounds the number of high degree neighbors a node has after $k$ scales of the algorithm.
In a sense the \textsc{Invariant} is trivially satisfied by design; nodes that do not satisfy the \textsc{Invariant} after $\Lambda$ iterations in Scale $k$
are simply placed in the set $B$ (of ``bad'' nodes) in Step 2(b) of the algorithm.
Of course, we have to later on deal with the nodes in $B$ somehow and so we cannot simply place all nodes in $B$ and claim to
have satisfied the \textsc{Invariant}!
Let $N$ denote the set on the left-hand side of the \textsc{Invariant} above.
The goal then is to show that, with probability at least $1 - 1/\Delta^2$, in Scale $k$, either $v$ becomes inactive (Lemma \ref{lemma:highN}) or
the size of $N$ falls to $\Delta/2^{k+2}$ or less (Lemma \ref{lemma:lowN}).
Showing this leads to Theorem \ref{theorem:lowProbability} which claims that that after Scale $k$, 
each active node satisfies the invariant with 
probability at least $1 - 1/\Delta^2$ and is therefore placed in $B$ with probability at most $1/\Delta^2$.

Unlike in the analysis of Algorithm \textsc{TreeIndependentSet},
for bounded arboricity graphs, the proof of Theorem \ref{theorem:lowProbability} has to deal 
with seemingly complicated dependencies among
probabilistic events that the algorithm depends on.
Our main contribution in this paper is to show that all of these dependencies can be quite naturally analyzed via
read-$k$ inequalities (with different values of the parameter $k$).
So first, in Section \ref{subsection:readKAction}, we use read-$k$ inequalities to analyze three key probabilistic events pertaining
to the progress of the algorithm. 
Later on we show how the success of these probabilistic events with sufficiently high probability holds the key to
proving Theorem \ref{theorem:lowProbability}.

\noindent
\textbf{Notation:} For the purposes of the analysis we fix an edge orientation of the given 
arboricity-$\alpha$ graph such that each node has at most $\alpha$ out-neighbors (parents).
We use $\parent_{IB}(v)$ to denote the set of currently active parents of node $v$ and
$\child_{IB}(v)$ to denote the set of currently active children of node $v$.
For any subset $S$ of nodes, we use $\Delta_{IB}(S)$ to denote $\max_{v \in S} deg_{IB}(v)$.

\subsection{Read-$k$ Inequalities in Action}
\label{subsection:readKAction}
In this section, we analyze via read-$k$ inequalities, three key probabilistic
events whose success (with sufficient probability) ensures rapid progress of 
our algorithm.
The first event concerns the interaction between nodes and their children and
the second concerns the interaction between nodes and their parents.
The third event is more complicated and it concerns the interaction between
nodes and their children, their children's children (i.e., \textit{grandchildren}) and their children's other parents (i.e., \textit{co-parents}).
To be more specific, let us fix a Scale $k$ and an iteration within that scale.
Let $M \subseteq V_{IB}$ be an active subset of nodes just before the start
of the iteration under consideration.
The three probabilistic events we analyze can be informally described
as follows.
For Events (1) and (2), we assume that all nodes in $M$ have degree at 
most $\rho_k$ and are therefore competitive.

\begin{framed}
\small{
\noindent
\textbf{Event (1)} Among the set of nodes $M$, there exists a node whose priority is larger
than the priority of all its children.\\
\textbf{Event (2)} Suppose that $M$ is sufficiently large. Then a large fraction of the nodes
in $M$ have priority greater than priorities of all their parents.\\
\textbf{Event (3)} Suppose that every node in $M$ has sufficiently high degree. Then a
large fraction of the nodes in $M$ become inactive due to their
children joining the MIS.}
\end{framed}

\begin{figure}[h]
\centering
\includegraphics[width = 0.85\textwidth]{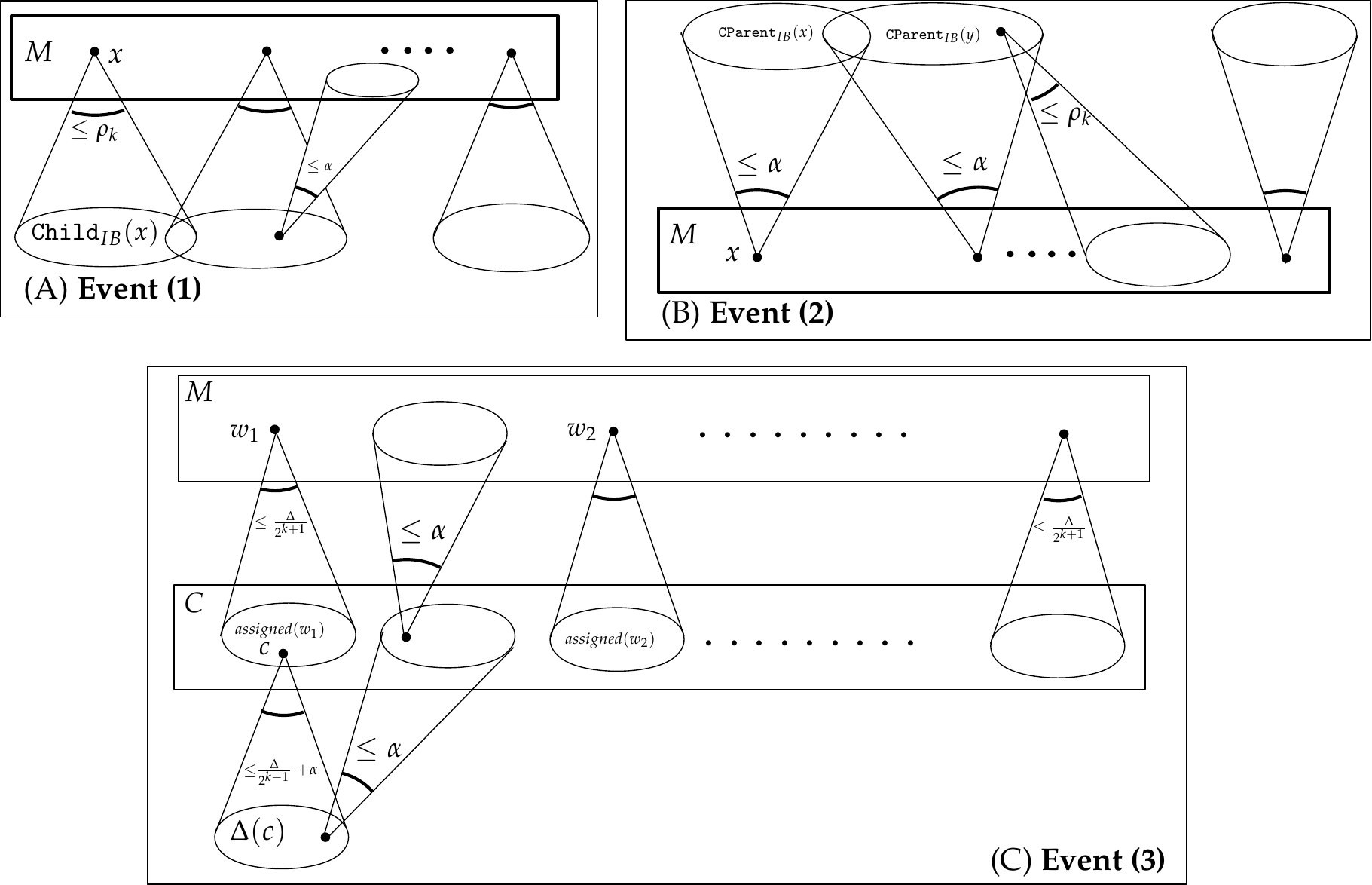}
\caption{(A) shows the application of read-$\alpha$ inequalities to lower bound the probability 
of some node $x \in M$ having priority greater than priorities of all its children.
(B) shows the application of a read-$\rho_k$ inequality to prove that with sufficient probability
a large fraction of the nodes in $M$ have priorities greater than priorities of parents.
(C) shows the application of a read-$\Theta(\alpha^2)$ inequality to prove that with
sufficient probability, a large fraction of the nodes in $M$ are eliminated by children joining the MIS.
}
\label{fig:readK}
\end{figure}

The simplest approach to analyzing these events is to decompose each event into 
sub-events centered at each of the nodes in $M$ and then apply
a tail inequality such as the Chernoff bound. 
The difficulty with this approach of course is the lack of independence among the sub-events at nodes in $M$.
However, as we discuss below and then show later, each of these collections of sub-events
can be analyzed using a read-$k$ inequality with different values of the parameter $k$.

Event (1) (see Theorem \ref{thm:readalpha} and Figure \ref{fig:readK}(A)) can be viewed as the complement of the event in which every
node in $M$ has a child with greater priority.
This latter event is a conjunction of events, $E_x$ for $x \in M$, where 
$E_x \equiv r(x) < \max\{r(y) \mid 
y \in \child_{IB}(x)\}$.
However, for nodes $x, x' \in M$, $E_x$ and $E_{x'}$ need not be independent because 
$x$ and $x'$ may share children.
Nevertheless, since a child can have at most $\alpha$ parents, the collection 
$\{E_x \mid x \in M\}$
of events has a dependency structure that forms a read-$\alpha$ family and
we can analyze Event (1) by applying the read-$\alpha$ conjunction inequality 
(Theorem \ref{theorem:conjunctionReadKFamily}).

We can attempt to analyze Event (2) (see Theorem \ref{thm:readrho} and Figure \ref{fig:readK}(B)) in a similar manner.
For each $x \in M$, let $F_x \equiv r(x) > \max\{r(y) \mid y \in \parent_{IB}(x)\}$.
However, dependencies among the events $\{F_x \mid x \in M\}$ are harder to deal with because 
a node can be the parent of arbitrarily many nodes in $M$ and thus possibly affect all nodes in $M$.
However, recall that a node with degree greater than $\rho_k$ does not participate in the 
competition to join the MIS (it simply sets its priority to 0).
Thus, if $M$ is significantly larger than $\rho_k$ then a competitive node can only be
the parent of a small fraction of nodes in $M$.
Thus the events $\{F_x \mid x \in M\}$ have a read-$\rho_k$ dependency structure 
and we can apply a read-$\rho_k$ tail inequality 
to analyze this event.

Event (3) (see Theorem \ref{thm:event3} and Figure \ref{fig:readK}(C)) pertains to the elimination of nodes in $M$ due to 
children of these nodes joining the MIS.
Following the approach used to analyze Events (1) and (2), we consider events
$G_x$ for $x \in M$ where $G_x$ is the event that some child of $x$ joins the MIS.
Whether a child $w$ of $x \in M$ joins the MIS, depends on the priorities at $w$ and neighbors
of $w$.
Specifically, $G_x$ depends on the priority of $x$ and the priorities of children of $x$, 
grandchildren of $x$, and co-parents of $x$.
As a result, the dependencies among the events $\{G_x \mid x \in M\}$ are much 
more complicated to analyze and cannot be directly analyzed using read-$k$ inequalities.
To get around this problem, we apply the analysis of Event (2) (Theorem \ref{thm:readrho})
to show that with sufficiently high probability, a substantial fraction of the children of 
$x \in M$ have priorities greater than all their parents.
We then condition on this event and only focus on such children (denoted $\child'_{IB}(x)$)
of each $x \in M$.
Now let us redefine $G_x$ as the event that some node in $\child'_{IB}(x)$ has priorities
greater than all its \textit{children}. Note that if a node $w \in \child'_{IB}(x)$ has priority
greater than its children, it will join the MIS (thereby eliminating $x$)
since its priority is known to be greater than the priorities of parents. 
Thus, if the redefined $G_x$ occurs, then $x$ is eliminated. Now note that each $G_x$
depends on the priority of $x$, priorities of children of $x$ and priorities of grandchildren
of $x$. 
Given that each node has at most $\alpha$ parents and $\alpha^2$ grandparents, we can see that the
collection $\{G_x \mid x \in M\}$ forms a read-$\alpha(\alpha + 1)$ family, allowing
us to use read-$\alpha(\alpha+1)$ inequalities to analyze Event (3).
In the three theorems that follow, we formally describe and analyze Events (1)-(3).

\begin{theorem} \label{thm:readalpha}
\textbf{Event (1)}
For some Scale $k$ and some iteration in this scale, let $M \subseteq V_{IB}$ be a subset of nodes that are active
just before the start of the iteration.
Further suppose that $\Delta_{IB}(M) \le \rho_k$.
Then, with probability at least $1 -  \left(1 - \frac{1}{\Delta_{IB}(M)}\right)^{|M|/(2\alpha^2)}$, some node in $M$ will choose a priority greater 
than the priorities of all of its children.
This holds even when we condition on all nodes in $M$ having priorities greater than their parents' priorities.
\end{theorem}
\vspace{-5mm}
\begin{proof}
Since the graph induced by $V_{IB}$ has arboricity at most $\alpha$, there exists an
independent set\footnote{By repeatedly
adding a vertex with degree at most $\alpha$ to the independent set, we can see that there is an independent set of size
at least $|M|/(\alpha+1)$ in the graph induced by $V_{IB}$.}
  $M_{ind} \subseteq M$ such that $\vert M_{ind} \vert \geq  \vert M \vert /2\alpha$.
Let $x^* \in M_{ind}$ be a node that chooses a priority
$r(x^*)$ greater than all its children, i.e., $r(x^*) > \max\{r(y) \mid y \in \child_{IB}(x^*)\}$ in the iteration being considered.
We now calculate the probability that such an $x^*$ exists.
For each node $x \in M_{ind}$, let $E_x$ denote the event $r(x) < \max\{r(y) \mid y \in \child_{IB}(x)\}$
and let $Y_x$ be the indicator variable for $E_x$. We now argue that the collection of random variables
$\{Y_x \mid x \in M_{ind}\}$ forms a read-$\alpha$ family.
See Figure \ref{fig:readK}(A).

\noindent
\textbf{Read-$\alpha$ family.}
Each $Y_x$ is a function of independent random variables, namely the priority $r(x)$
and the priorities of children of $x$, i.e., $\{r(y) | y \in \child_{IB}(x)\}$.
Thus a priority $r(w)$ can only influence random variables $Y_x$, where $x$ is a parent of $w$
and this means that each priority can influence at most $\alpha$ elements in
$\{Y_x \mid x \in M_{ind}\}$.
Therefore the set of random variables $\{Y_x \mid x \in M_{ind}\}$ forms a read-$\alpha$ family.

Now note that $Y_x = 0$ corresponds to $r(x)$ being larger than $r(y)$ for all $y \in \child_{IB}(x)$.
Therefore, $Pr(Y_x = 0) \ge \frac{1}{\Delta_{IB}(M)}$, implying that 
$Pr(Y_x = 1) \le \left(1-\frac{1}{\Delta_{IB}(M)}\right)$.
Note that this depends on the fact that $deg_{IB}(x) \le \rho_k$ and $x$ is competitive.
Using this bound and the conjunctive read-$\alpha$ inequality in
Theorem \ref{theorem:conjunctionReadKFamily}, we see that 
$Pr(\cap_{x\in M_{ind}} Y_x = 1) \leq (1-1/\Delta_{IB}(M))^{(|M|/2\alpha) \cdot (1/\alpha)}$.
Thus the probability that there exists an $x^* \in M_{ind}$ for which $E_{x^*}$ holds 
is as claimed.
\end{proof}

\begin{theorem}\label{thm:readrho}
\textbf{Event (2)}
For some scale $k$ and some iteration in this scale, let $M \subseteq V_{IB}$ be a subset of nodes that are active
just before the start of the iteration.
Further suppose that $\Delta(M) \le \rho_k$ and 
$|M| > 64 \alpha^2 \cdot \ln^2 \Delta \cdot \Delta/2^{k+1}$. Then, at the end of the iteration, 
with probability at least $(1-1/\Delta^4)$, the number of nodes in $M$ that choose a priority greater than their parents is more than $\vert M \vert /2\alpha$.
\end{theorem}
\vspace{-5mm}
\begin{proof}
The probability that a node in $M$ chooses a priority greater than its parents is equal to the probability that it chooses a priority greater than 
its competitive parents. (Recall that a non-compete node has degree more than $\rho_k$ and it deterministically sets its priority to 0.) 
Let $\texttt{CParent}_{IB}(u)$ denote the set of current compete parents of a node $u$.\\
For any node $u \in M$, let $F_u$ denote the event
$r(u) > \max(\{r(y)| y \in \texttt{CParent}_{IB}(u)\})$ and let $X_u$ be the indicator variable for $F_u$.
Let $X = \sum_{u \in M} X_u$ be the random variable representing the number of nodes in $M$ whose priorities are greater than
priorities of their parents.
Since each node can have at most $\alpha$ parents and since 
$deg_{IB}(x) \le \rho_k$, $Pr(X_u=1) = E[X_u] \geq 1/\alpha$ and $E[X] \geq \vert M \vert/\alpha$. 
We would now like to show that $X$ is concentrated about its expectation, but cannot use Chernoff bounds because
the variables $\{X_u \mid u \in M\}$ are not mutually independent.
Again, a read-$k$ inequality comes to the rescue and we first show that the set of variables
$\{X_u \mid u \in M\}$ forms a read-$\rho_k$ family.

\noindent
\textbf{Read-$\rho_k$ family.}
Each $X_u$ is a function of independent random variables, namely its own priority $r(u)$ and the priorities of its competitive parents. 
Since any competitive node $w \in V_{IB}$ has degree at most $ \rho_k$, a priority $r(w)$ influences at most $\rho_k$ $X_{u}$'s. 
Therefore, $\{X_u \mid u \in M\}$ forms a read-$\rho_k$ family and we can apply 
the read-$\rho_k$ tail inequality in Theorem \ref{thm:readk} (Form (1)) 
to establish the concentration of $X$ about its expectation as follows:
$$Pr(X \leq (1/\alpha - 1/2\alpha) \cdot \vert M \vert) \leq \exp\left( -2(1/4\alpha^2) \cdot \frac {\vert M \vert}{\rho_k} \right ).$$ 
Since $|M| > 64\alpha^2 \ln^2\Delta \cdot \Delta/2^{k+1}$,
$$Pr(X \leq\vert M \vert /2\alpha) \leq  \exp\left( -2(1/4\alpha^2) \cdot  \frac{ \Delta \alpha^2 (64 \ln^2 \Delta)/2^{k+1})} {\Delta (8\ln \Delta)/2^{k+1} } \right) \leq \exp(-4 \ln\Delta).$$ 
Thus, the probability that $X > |M|/2\alpha$ is at least $(1-1/\Delta^{4})$.
\end{proof}

\begin{theorem}\label{thm:event3}
\textbf{Event (3)}
For some scale $k$ and some iteration in this scale, let $M \subseteq V_{IB}$ be a subset of 
nodes that are active just before the start of the iteration.
Further suppose that $|M| > \Delta/2^{k+2}$ and 
$\deg_{IB}(w) > \Delta/2^{k} + \alpha$ for all nodes $w \in M$. 
Then with probability at least $(1 - 1/\Delta^3)$ at least $|M|/8\alpha^2(32\alpha^6 + 1)$ nodes in
$M$ are eliminated in the iteration.
\end{theorem}
\vspace{-5mm}
\begin{proof}
Applying the \INVARIANT, 
at the end of the scale $k-1$, 
we see that each node $w$ in $M$ has at most $\Delta/2^{k+1}$ neighbors 
with degree more than $\Delta/2^{k-1} + \alpha$. 
Therefore, $w$ has at least $\deg_{IB}(w) - \Delta/2^{k+1} -\alpha > \Delta/2^{k+1} $ children with degree at most $\Delta/2^{k-1} +\alpha$. 
For the purposes of this theorem, we will refer to these nodes as 
\textit{low-degree} children. 

We now construct a set $C$, that consists of low-degree children of nodes in 
$M$.  
Consider nodes in $M$ in some arbitrary order $w_1, w_2, \ldots, w_{|M|}$. 
For $w_1$, pick $\Delta/2^{k+1}$ low-degree children from among the more than 
$\Delta/2^{k+1}$ such children that it has. 
These nodes are said to be \textit{covered} and \textit{assigned} to $w_1$. 
For each node $w_i$, $1< i \leq |M|$, let $c_i$ be the number of low-degree 
children of $w_i$ that have already been covered. 
If $c_i$ is at least $\Delta/2^{k+1}$, we do nothing. 
Otherwise, pick $(\Delta/2^{k+1} - c_i)$ low-degree children of $w_i$ 
arbitrarily and declare these nodes covered and assign them to $w_i$. 
Let $C$ be the set of all covered nodes at the end of this procedure. 

Now note that each node in $w_i \in M$ has at least $\Delta/2^{k+1}$ 
children in $C$ and at most $\Delta/2^{k+1}$ of these children are assigned 
to it. 
See Figure \ref{fig:readK}(C).
Since each node in $C$ has at most $\alpha$ parents in $M$, $C$ has size at 
least $\frac{|M|}{\alpha} \cdot \frac{\Delta}{2^{k+1}}$. 
Note that since $|M| > \Delta/2^{k+2}$ and the maximum value of the scale
index $k$ is bounded above by 
$\log \left (\frac{\Delta}{1176\cdot 16\alpha^{10}\ln^2\Delta} \right )$, 
using a little algebra we see that 
$|M|$ is more than $64 \alpha^4 \ln^2 \Delta$ and therefore
$|C|$ is more than $64 \alpha^3 \ln^2 \Delta \cdot \Delta/2^{k+1}$ for all 
values of $k$. 
Then, applying Theorem \ref{thm:readrho} on the set $C$ (since it is large enough), we see that 
with probability at least $1-1/\Delta^4$, more than $|C|/2\alpha$ nodes in $C$ 
choose a priority higher than their parents' priority. 
Let $C'$ denote the subset of nodes in $C$ that have chosen a 
priority higher than priorities of their parents.
Let $E$ denote the event that $|C'| > |C|/2\alpha$. 
(Thus, $E$ happens with probability at least $1 - 1/\Delta^4$.)
We now condition on event $E$ and using a simple averaging argument  we show that
there are a significant fraction of the nodes in $M$,
each having sufficiently many children in $C'$.
This is stated in the claim below.
The point of this is that for such nodes in $M$ to be eliminated, it would
suffice for a child in $C'$ to have priority larger than priority of
its children -- since nodes in $C'$ already have priority more than priorities
of parents.\\

\noindent\textbf{Claim:}\textit{ Conditioned on $E$, there are at least $|M|/4\alpha^2$ nodes in $M$ that have more than 
$\frac{1}{2\alpha^3}\cdot \frac{\Delta}{2^{k+1}}$ children each in $C'$.}
\begin{proof}
Let $O$ be the subset of nodes in $M$ that have at most
$\frac{1}{2\alpha^3} \cdot \frac{\Delta}{2^{k+1}}$ children in $C'$.
To calculate a lower bound on $|M|-|O|$, we will try to cover nodes in $C'$ using $O$ and $M \setminus O$.
Each node in $O$ is assigned at most $\frac{1}{2\alpha^3} \cdot \frac{\Delta}{2^{k+1}}$ nodes in $C'$ and each node in $M \setminus O$ is assigned at most
$\Delta/2^{k+1}$ nodes in $C'$.
Thus,
$$(|M|-|O|)\cdot \frac{\Delta}{2^{k+1}} + |O|\left(\frac{1}{2\alpha^3}\cdot \frac{\Delta}{2^{k+1}}\right) \geq |C'| \geq \frac{|C|}{2\alpha} \geq \frac{|M|}{2\alpha^2}\cdot \frac{\Delta}{2^{k+1}}.
$$
Note that the second-last inequality above depends on the conditioning on event $E$.
Manipulating this expression we get the following upper bound on $|O|$:
$$|O| \leq |M|\left(\frac{1-1/2\alpha^2}{1-1/2\alpha^3}\right).$$
Therefore,
$$|M|-|O| \geq |M|\left(1 -  \frac{1-1/2\alpha^2}{1-1/2\alpha^3}\right) 
                \geq \frac{|M|}{4\alpha^2}.$$

\noindent
The last inequality above holds for all $\alpha \ge 2$.
\end{proof}

Let $M'$ denote the subset of $M$ of nodes each having at least 
$\frac{1}{2\alpha^3} \cdot \frac{\Delta}{2^{k+1}}$ children in $C'$.
Thus the above claim shows that conditioned on event $E$, $|M'| \ge |M|/4\alpha^2$.
Consider an arbitrary node $w \in M'$.
Now note that $|\child_{IB}(w) \cap C'| \ge \frac{1}{2\alpha^3} \cdot \frac{\Delta}{2^{k+1}}$
and
$\Delta(\child_{IB}(w) \cap C') \le \Delta/2^{k-1} + \alpha$.
This means that we can apply Theorem \ref{thm:readalpha} to the set 
$\child_{IB}(w) \cap C'$ and conclude that the probability that some node in
$\child_{IB}(w) \cap C'$ will have priority greater than the priorities of
all its children is at least

$$1 - \left(1-\frac{1}{\Delta/2^{k-1}+\alpha}\right)^{\Delta/(2\alpha^3 \cdot 2^{k+1}) \cdot (1/2\alpha^2)} \geq 1-\exp\left(-\frac{2^{k-1}}{2\Delta \alpha} \cdot \frac{\Delta}{2^{k+1}4\alpha^5}\right) 
> \left(1-e^{-1/32\alpha^6}\right).$$
This last expression can be bounded below by $1/(32\alpha^6+1)$.

For any $w \in M'$, let $G_w$ denote the event that some node $\child_{IB}(w) \cap C'$
has priority greater than the priorities of all its children. Let $Z_w$ be the indicator
variable for event $G_w$. By the above calculation we see that $Pr(Z_w = 1) \ge \frac{1}{32\alpha^6+1}.$
Let $Z = \sum_{w \in M'} Z_w$. 
Note that if a node $x$ in $\child_{IB}(w) \cap C'$ has priority greater than the priorities
of children, then it joins the MIS since we already know that it has priority greater than the 
priorities of parents.
Thus $Z$ is a lower bound on the number of nodes in $M'$ that are eliminated in
this iteration of the algorithm.
By linearity of expectation, we see that $E[Z] \ge \frac{|M'|}{32\alpha^6+1}$. 
We would now like to finish the proof of the theorem by showing that with sufficiently 
high probability, $Z$ is at least one-half of its expectation.
Unfortunately, the $Z_w$'s are not mutually independent and we cannot use Chernoff tail bounds
to show the concentration of $Z$ about its expectation.
Nevertheless we are able to show that the random variables $\{Z_w \mid w \in M'\}$
form a read-$\alpha(\alpha + 1)$ family and exploit this structure to show the tail bound we need.

\noindent
\textbf{Read-$\alpha(\alpha+1)$ family.}
Note that each $Z_w$ is a function of $r(w)$, priorities of children of $w$, and
priorities of grandchildren of $w$. 
It is important to note here that parents of $w$ and co-parents of $w$ have no role
to play in determining the value of $Z_w$.
Since the graph has arboricity $\alpha$, for any node $x$, $r(x)$ may influence
at most $\alpha(\alpha + 1)$ of the variables in $\{Z_w \mid w \in M'\}$.
Using the read-$\alpha(\alpha + 1)$ tail inequality in Theorem \ref{thm:readk} (Form (2)), 
we see that:
$$
Pr(Z < E[Z]/2)  \leq \exp\left(-\frac{E[Z]}{8\alpha(\alpha+1)}\right) 
                \leq \exp\left(-\frac{|M'|}{8\alpha(\alpha+1)(32\alpha^6+1)}\right).
$$
Now we condition on the event $E$ and use the fact that conditioned on $E$,
$|M'| \ge |M|/4\alpha^2$ and $E[Z] \ge |M|/(4\alpha^2(32\alpha^6 + 1))$ to obtain:
$$
Pr\left(Z < \frac{|M|}{8\alpha^2(32\alpha^6+1)} \bigg{|} E\right)  \leq \exp\left(-\frac{|M|}{32\alpha^3(\alpha+1)(32\alpha^6+1)}\right). 
$$
According to the hypothesis of the theorem, $|M| > \Delta/2^{k+2}$ and we know that
the maximum value of the scale index $k$ is bounded above by
$\log \left (\frac{\Delta}{1176\cdot 16\alpha^{10}\ln^2\Delta} \right )$.
Using a little algebra we see that
$|M|$ is more than $1176 \cdot 4 \alpha^{10} \ln^2 \Delta$ for all values of $k$.
Therefore,
$$
Pr\left(Z < \frac{|M|}{8\alpha^2(32\alpha^6+1)} \bigg{|} E\right)  \leq \exp\left(-\frac{1176 \cdot 4 \alpha^{10} \ln^2 \Delta}{32\alpha^3(\alpha+1)(32\alpha^6+1)}\right) \leq \exp(-\ln^2 \Delta). 
$$
Finally, noting that $Pr(E) \ge (1 - 1/\Delta^4)$, we see that
$$Pr\left(Z \ge \frac{|M|}{8\alpha^2(32\alpha^6+1)}\right) \geq (1 - \exp(-\ln^2 \Delta)) \cdot (1 - 1/\Delta^4) \ge (1 - 1/\Delta^3).$$
Therefore, with probability at least $(1-1/\Delta^3)$, at least $|M|/8\alpha^2(32\alpha^6 + 1)$ fraction of the nodes in $M$ are eliminated in each iteration.
\end{proof}

\subsection{Proving the Invariant}
In this section we show inductively that the \INVARIANT\ holds after every 
scale. Suppose that the \INVARIANT\ holds after Scale $k-1$ for any $k \ge 1$.
(Note that ``end of Scale 0'' refers to the beginning of the algorithm.)
Fix a node $v$ and let 
$N = \{ w \in \Gamma_{IB}(v) \vert \deg_{IB}(w) > \Delta/2^{k} + \alpha \}$ 
be the set of active \textit{high-degree} neighbors of $v$ at the beginning of Scale $k$.
To establish that the \INVARIANT\ holds after Scale $k$ we will show that with sufficiently
high probability either (i) $v$ is eliminated in Scale $k$ or (ii) $|N|$ shrinks
to at most $\Delta/2^{k+2}$ by the end of Scale $k$.
We consider two cases depending on the size of $N$ and show that (i) holds when
$|N|$ is large (Lemma \ref{lemma:highN}) and (ii) holds when $|N|$ is smaller, but still 
bigger than $\Delta/2^{k+2}$ (Lemma \ref{lemma:lowN}).
We note that this organizational structure of our overall proof is similar to the 
approach used by Barenboim et al.~ \cite{BEPS12arxiv,BEPS12FOCS}. 
Our main innovation and contribution appears in the
previous section where we analyze, via read-$k$ inequalities, key probabilistic events that
Lemmas \ref{lemma:highN} and \ref{lemma:lowN} depend on.

We first briefly describe the role Events (1)-(3) (Section \ref{subsection:readKAction})
play in the proofs of these lemmas.
Applying the \INVARIANT\ after scale $k-1$ to $v$ implies that a large number of nodes in $N$ have degree at most $\Delta/2^{k-1} + \alpha$.
This set of ``low degree'' nodes is large enough for us to consider Event (2) at these nodes and using Theorem \ref{thm:readrho} 
we see that a large fraction of these nodes have priority greater than their parents (with sufficiently high probability).
Conditioning on this event, we then consider Event (1) at the ``low degree'' nodes in $N$  
whose priorities are larger than priorities of parents.
We then apply Theorem \ref{thm:readalpha} to conclude that with probability at least $1 - 1/\Delta^2$ at least one 
node in $N$ joins the MIS, thereby eliminating $v$ and yielding Lemma \ref{lemma:highN}.
To obtain Lemma \ref{lemma:lowN}, we repeatedly consider Event (3) at the nodes in $N$
and apply Theorem \ref{thm:event3} to obtain a decay of roughly $1/\alpha^8$ fraction, after each iteration 
with sufficiently high probability.
Performing $\Lambda = \Theta(\alpha^8(\log \alpha + \log\log \Delta)$ iterations is enough to reduce 
$|N|$ to at most $\Delta/2^{k+2}$ with sufficiently high probability.
Lemmas \ref{lemma:highN} and \ref{lemma:lowN} immediately lead to Theorem \ref{theorem:lowProbability}.

\begin{lemma}\label{lemma:highN}
If $|N| > 130 \alpha^4 \cdot \ln^2 \Delta \cdot \Delta/2^{k+1}$ for $\Lambda/p$ iterations in scale $k$, then $v$ is eliminated with
probability at least $1-1/\Delta^2$.
\end{lemma}
\begin{proof}
We focus on the first iteration of Scale $k$.
By applying the \INVARIANT\ at the end of Scale $k-1$ to node $v$, we see that $v$ has at most
$\Delta/2^{k+1}$ neighbors with degree more than $\Delta/2^{k-1} + \alpha$.
Thus among the nodes in $N$, there are at least
$$|N| - \Delta/2^{k+1} > 130 \alpha^4 \cdot \ln^2 \Delta \cdot \Delta/2^{k+1} - \Delta/2^{k+1} > 
129 \alpha^4 \cdot \ln^2 \Delta \cdot \Delta/2^{k+1}$$
with degree at most $\Delta/2^{k-1} + \alpha$.
Let $N_{low} \subseteq N$ denote the subset of $N$ of nodes with degree at most
$\Delta/2^{k-1} + \alpha$.
(Thus, we have just established that $|N_{low}| > 129 \alpha^4 \cdot \ln^2 \Delta \cdot \Delta/2^{k+1}$.)
Since $N_{low}$ is large enough, we can apply Theorem \ref{thm:readrho} to conclude that
with probability at least $1 - 1/\Delta^4$, at least $|N_{low}|/2\alpha$ nodes in
$N_{low}$ have priorities that are larger than priorities of their parents.
(Recall that this is Event (2) at $N_{low}$.)
Call this event $E_{par}$ and let $N_{par} \subseteq N_{low}$ denote the subset of
nodes in $N_{low}$ whose priorities are larger than priorities of their parents.
Thus, if we condition on $E_{par}$, we get that $|N_{par}| \ge |N_{low}|/2\alpha$.
We now apply Theorem \ref{thm:readalpha} to the set $N_{par}$ to get a lower bound on
the probability
that $N_{par}$ contains a node whose priority is greater than the priorities
of all children. Letting $F$ denote this event, we get the lower bound:
$$Pr(F) \ge 1 - \left(1 - \frac{1}{\Delta(N_{par})}\right)^{|N_{par}|/(2\alpha^2)}.$$
Since $N_{par} \subseteq N_{low}$, we know that $\Delta(N_{par}) \le \Delta/2^{k-1} + \alpha$
and if we condition on $E_{par}$, we know that $|N_{par}| > |N_{low}|/2\alpha > 64\alpha^3 \cdot
\ln^2 \Delta \cdot \Delta/2^{k+1}$.

$$Pr(F | E_{par})  \geq  1 - \left(1-\frac{1}{\Delta/2^{k-1} + \alpha}\right)^{\frac{|N_{par}|}{2\alpha^2}}
\geq 1 - \exp\left(-\frac{2^{k-1}}{2\Delta \alpha} \cdot 32 \alpha \ln^2 \Delta \cdot \frac{\Delta}{2^{k+1}} \right)\\
\geq 1 - 1/\Delta^4.$$
Since $E_{par}$ occurs with probability at least  $1-1/\Delta^4$,
we conclude that
\begin{align*}
Pr(F) =  Pr(E | E_{par})\cdot Pr(E_{par}) &\geq (1-1/\Delta^{4})(1-1/\Delta^{4}) \geq 
(1-1/\Delta^2).\\
\end{align*}
\end{proof}

\begin{lemma}\label{lemma:lowN}
If $|N| \leq 130\alpha^4 \cdot \ln^2 \Delta \cdot \Delta/2^{k+1}$ (at the beginning of Scale $k$), 
then after all $\Lambda/p$ iterations of Scale $k$, $|N| \leq \Delta/2^{k+2}$ with probability at least $1-1/\Delta^2$.
\end{lemma} 
\begin{proof}
Let $n_i$ denote the size of $N$ before iteration $i$, $1 \le i \le \Lambda/p$, in Scale $k$
and let $n_{\Lambda/p+1}$ denote the size of $N$ after the $\Lambda/p+1$ iteration in Scale $k$.
Suppose that $n_{\Lambda/p+1} > \Delta/2^{k+2}$.
Then, $n_i > \Delta/2^{k+2}$ for all $i$, $1 \le i \le \Lambda$ and so we can appeal to
Theorem \ref{thm:event3} and conclude that for all $i \in [\Lambda/p]$
$$\Pr\left(n_{i+1} \le \left(1 - \frac{1}{8\alpha^2(32\alpha^6 + 1)}\right) \cdot n_i\right) \ge 1 - 1/\Delta^3.$$
By the union bound, 
$$
Pr\left(\text{There exists }i: n_{i+1} > \left(1 - \frac{1}{8\alpha^2(32\alpha^6 + 1)}\right)\right) \leq \frac{\Lambda/p}{\Delta^3}\leq \frac{1}{\Delta^2}.
$$
In other words, with probability at least $1 - 1/\Delta^2$, $n_{i+1} \le (1 - 1/(2(32\alpha^6+1))) \cdot n_i$ for  $\Lambda/p$ iterations.
Therefore, with probability at least $1 - 1/\Delta^2$,

$$n_{\Lambda+1} \le \left(1 - \frac{1}{8\alpha^2(32\alpha^6 + 1)}\right)^{\Lambda/p} \cdot n_1.$$
We now observe that 
$$\left(1 - \frac{1}{8\alpha^2(32\alpha^6 + 1)}\right)^{\Lambda/p} \le \exp\left(-\frac{1}{8\alpha^2(32\alpha^6 + 1)} \cdot \frac{\Lambda}{p}\right).$$

This implies that choosing $\Lambda$ at least $p \cdot 8\alpha^2 (32\alpha^6 + 1) \cdot \ln(260 \alpha^4 \ln^2 \Delta)$
suffices to guarantee that
$$\left(1 - \frac{1}{8\alpha^2(32\alpha^6 + 1)}\right)^{\Lambda/p} \cdot n_1 \le \frac{n_1}{260 \alpha^4 \ln^2 \Delta}.$$
Now note that we choose $\Lambda = \lceil p \cdot 8\alpha^2 (32\alpha^6 + 1) \cdot \ln(260 \alpha^4 \ln^2 \Delta) \rceil$ in 
Algorithm \textsc{BoundedArbIndependentSet}.
Since $n_1 \le 130\alpha^4 \cdot \ln^2 \Delta \cdot \Delta/2^{k+1}$,
it follows that with probability at least $1 - 1/\Delta^2$,
$|N| \le \Delta/2^{k+2}$ after $\Lambda/p$ iterations of scale $k$.
\end{proof}

\begin{theorem} 
\label{theorem:lowProbability}
In any Scale $k$, a node $v$ that is in $V_{IB}$ at the start of the Scale is included in $B$ with 
probability at most $1/\Delta^{2p}$, independent of random choices of nodes outside its three neighborhood.  
\end{theorem}
\begin{proof}
We look at $\Lambda$ iterations of each scale in chunks of $\Lambda/p$ iterations. As a direct consequence of lemmas \ref{lemma:lowN} and \ref{lemma:highN}, a node $v$ violates the invariant with probability at most $1/\Delta^2$ after a chunk. The probability that a node is bad at the end of the scale, is equal to the probability that its bad at the end of each chunk of $\Lambda/p$ iterations. Thus, the probability that a node is bad at the end of the scale is at most $(1/\Delta^2)^p = 1/\Delta^{2p}$. 

We now argue that this probability, for each node, is independent of nodes outside its three-neighborhood. 
A node $v$ joins $B$ if it violates the invariant at the end of a scale. This means that a lot of high degree neighbors of $v$ survive the scale. The survival of these high degree nodes depends on their neighbors ($v$'s two-neighborhood) joining the MIS. The event that nodes in $v$'s two-neighborhood join the MIS, in turn, depends on these nodes choosing higher priorities than their neighbors, which can be at most three hops away from $v$. Thus $v$ joins the bad set with probability at most $1/\Delta^{2p}$, independent of nodes outside its three-neighborhood.
\end{proof}
\subsection{Finishing Up the MIS Computation}
\label{section:finishingUP}

Theorem \ref{theorem:lowProbability}, which shows that every node joins $B$ with probability at most $1/\Delta^{2p}$, has the following
immediate consequence, shown in \cite{BEPS12arxiv,BEPS12FOCS}.
\begin{lemma}
\label{lemma:smallConnectedComponents}
All connected components in the subgraph induced by $B$ have at most $(\Delta^6  \cdot c \log_\Delta n)$ nodes with probability at least $1-n^{-c}$
\end{lemma}
\begin{proof}
Let $t := c\log_\Delta n $. Let $G^{[7,13]}$ denote the graph in which we put an edge between any two nodes $u,v$ with $dist(u,v) \in [7,13]$ in $G$. Consider a set of nodes $U \in V_{B}$ such that $U$ forms a connected component in $G^{[7,13]}$. Notice that such a $U$ implies the existence of a $t$-node tree in $G[V_B]$. We first show that with probability at least $1-1/n^{c}$, no $U$ with $|U| > t$ exists.\\
We want to argue that the probability $v$ join $B$ in scale is independent of any other node joining $B$ if it is at least at distance $7$ from $v$. If $v$ becomes bad at the end of scale $k$, then it must mean that $v$ violates the invariant at the end of scale $k$ i.e., it has a large number of high degree neighbors. Consequently, we can fix nodes in the three neighborhood of $v$ that survive all scales of the algorithm until $v$ joins $B$. Since these nodes are guaranteed to survive all scales until $k$ of the algorithm, the event that $v$ joins $B$ is independent of nodes having joined $B$ in all other scales, as long as $v$ doesn't share its three neighborhood with these nodes. Using Theorem \ref{theorem:lowProbability}, this probability is at most $1/\Delta^{2p}$ in each scale, and by fixing nodes, its independent of what happens outside its three neighborhood over all scales. 

Notice that two nodes in $U$ have distance at least $7$ from each other. Therefore, the probability that nodes in $U$ have all joined the bad set $B$, in any scale of the algorithm, is independent of each other. However, each node may have joined the set $B$ in any of at most $\log\Delta$ scales. Consequently, using Theorem \ref{theorem:lowProbability}, the probability that a $U$ with $|U|>t$ exists is at most $(\log\Delta)^t \cdot \Delta^{-2pt}$.
There are at most $4^t$ distinct topologies for rooted $t$-node trees and at most $n \Delta^{13t}$ ways to embed such a tree in the $G^{[7,13]}$. By the union bound, the probability that a cluster $U$, with size at least $t$, exists is:\\ 
\begin{align*}	
\leq 4^t \cdot n \Delta^{13t} \cdot (\log\Delta)^t \cdot \Delta^{-2p t}  \leq n^{-c}
\end{align*}
for large $\Delta$ and appropriate value of the parameter $p$. Now, consider a connected component $B \in V$ with more than $t \cdot \Delta^6$ nodes. Form the cluster $U$ by adding an arbitrary node, $u$, from $B$ to $U$, and remove less than $\Delta^4$ nodes that are within distance $6$ of $u$, from consideration in $U$. Then repeat the process until there are no nodes left in $B$ for consideration in $U$. Thus, any connected component in $B$ with size at least $\Delta^6 \cdot t $ contains a set $U$ with size at least $t$. Therefore, all connected components in $B$ have size $O(\Delta^6 \cdot \log_{\Delta} n)$ with probability at least $1-n^{-c}$.
\end{proof}

We now describe and analyze an algorithm, we call \textsc{ArbMIS} that takes the output of \textsc{BoundedArbIndependentSet} 
(Section \ref{section:algorithm}) and completes the computation of an MIS.
Recall that after the termination of \textsc{BoundedArbIndependentSet}, we get three (disjoint) sets of nodes, $V_{IB}$, $I$, and $B$ with 
the following properties:
\begin{enumerate}
\item[(i)] $I$ is an independent set.
\item[(ii)] Applying the \INVARIANT\ at the end of Scale 
$\Theta = \left\lfloor\log\left(\frac{\Delta}{1176\cdot 16\alpha^{10}\ln^2\Delta}\right)\right\rfloor$, we see that no
node in $V_{IB}$ has more than $1176\cdot 4 \alpha^{10} \ln^2 \Delta$ neighbors with degree more than 
$1176 \cdot 16 \alpha^{10} \ln^2 \Delta + \alpha$.
\item[(iii)] All connected components in the graph induced by $B$, have size less than $O(\Delta^6 \cdot c \log_\Delta n)$ with probability $1-n^{-c+1}$.
\end{enumerate}

\noindent
In the next step, divide $V_{IB}$ into the sets 
$V_{lo} = \{v \in V_{IB} | \deg_{IB}(v) \leq 1176 \cdot 16 \alpha^{10} \ln^2 \Delta + \alpha\}$ and 
$V_{hi} = \{v \in V_{IB} | \deg_{IB}(v) > 1176 \cdot 16 \alpha^{10} \ln^2 \Delta + \alpha\} $.  
By definition, $G[V_{lo}]$ has maximum degree  $1176 \cdot 16 \alpha^{10} \ln^2 \Delta + \alpha$ and by Property (ii) above, 
$G[V_{hi}]$ has maximum degree $1176\cdot 4 \alpha^{10} \ln^2 \Delta$. 
There are various options for computing an MIS on a arboricity-$\alpha$ graph with bounded degree.
We use an algorithm from Barenboim et al.~(Theorem 7.4, \cite{BEPS12arxiv}) to compute an MIS of $G[V_{lo}]$ in
$O(\log\log n \cdot \log \alpha + (\log \log \Delta)^2 + \alpha^2)$ time in the $\mathcal{CONGEST}$ model. 
Subsequently, we use the same algorithm on $V_{hi}\setminus \Gamma(I_{lo})$ to get an independent set 
$I_{hi}$ in the same time. 
Let $\eta$ be the set of all connected components in $B$. 
The following lemma shows that an MIS can be computed efficiently for each connected component
in the graph induced by $B$.

\begin{lemma}\label{lemma:cc} For each connected component in $B$, an MIS can be computed in time at most $O(\log \Delta + \log\log n + \alpha\log^* n)$ time, using messages of size at most $O(\log n)$.
\end{lemma}
\begin{proof}
The Barenboim-Elkin forest decomposition algorithm \cite{BE08} computes a $4\alpha$ forest decomposition in  $O(\log t)$ rounds using messages of size at most $O(\log t)$, on graphs with arboricity $\alpha$ and size $t$. Moreover, this forest decomposition algorithm also gives an orientation of edges for each forest.
Since each connected component in $B$ has size at most $O(\Delta^6 \log_\Delta n)$, we can compute a $4\alpha $ forest decomposition in parallel for each connected component in $G[B]$, in $O(\log \Delta + \log \log n)$ rounds, and get an orientation of edges for each forest.

Given an orientation of the edges, the Cole-Vishkin\cite{ColeVishkin} deterministic MIS algorithm computes a coloring in $O(\log^* t)$ rounds, using constant message sizes, for trees of size $t$. Consequently, nodes in each forest computed by the Barenboim-Elkin algorithm, can compute an MIS using the Cole-Vishkin algorithm in turn. Each time an MIS is computed for one forest, nodes remove themselves from consideration in the MIS, before the computation of an MIS for the next forest. Thus, in $O(\alpha \cdot \log^* n)$ time, an MIS can be computed for all the forests, computed as a result of the Barenboim-Elkin algorithm.  
\end{proof}

The total run-time of this algorithm is $O\left(\alpha^8 \cdot \log (\alpha \log \Delta) \cdot \log \Delta + \log\log n \cdot \log \alpha\right)$.
If $\Delta > \alpha \cdot 2^{\sqrt{\log n/\log\log n}}$, 
use the independent set algorithm from \cite{BEPS12arxiv,BEPS12FOCS} to reduce the max degree to 
$\alpha \cdot 2^{\sqrt{\log n \cdot \log\log n}}$ in $O(\sqrt{\log n \cdot \log\log n})$ time, using messages of size $O(\log n)$.  
Then, apply \textsc{ArbMIS} to compute an MIS for a total run time of $O(\alpha^9 \cdot \sqrt{\log n \cdot \log\log n})$. We note that this entire algorithm uses messages of size at most $O(\log n)$ i.e., it runs in the $\mathcal{CONGEST}$ model. Our finishing up technique is general, and can be applied to sparse graphs in other graph shattering algorithms as well. For example, coupled with the recent results of Ghaffari \cite{GhaffariSODA16}, this gives an $O(\log \alpha + \sqrt{\log n})$ run-time for computation of MIS for bounded arboricity graphs in the $\mathcal{CONGEST}$ model.  \\

\noindent
\RestyleAlgo{boxruled}
\begin{algorithm}[H]
\caption{\textsc{ArbMIS(Graph G)}:}
\begin{algorithmic}[1]
\State $(I,B) \leftarrow \textsc{BoundedArbIndependentSet}(G)$ \;
\State Partition $V_{IB} = V(G)\setminus (I \cup \Gamma(I) \cup B)$ into \;
$V_{lo} = \{v \in V_{IB} | \deg_{IB}(v) \leq 1176 \cdot 16 \alpha^{10} \ln^2 \Delta + \alpha\}$ \DontPrintSemicolon \;
$V_{hi} = \{v \in V_{IB} | \deg_{IB}(v) > 1176 \cdot 16 \alpha^{10} \ln^2 \Delta + \alpha\}$ \DontPrintSemicolon \;
\State Compute maximal independent sets on $V_{lo}$ and $V_{hi}$ \;
$I_{lo} \leftarrow$ an MIS of the graph induced by $V_{lo}$ \DontPrintSemicolon \;
$I_{ho} \leftarrow$ an MIS of the graph induced by $V_{hi}\setminus \Gamma(I_{lo})$ \DontPrintSemicolon \;
\State Let $\eta$ be a set of all connected components in $B$. For each $C$ in $\eta$  \;
$I_C \leftarrow$ an MIS of C
\State return $I \cup I_{lo} \cup I_{hi} \cup \bigcup_{c \in \eta} I_c$
\end{algorithmic}
\end{algorithm}

\newpage

\bibliography{arbmis}

\end{document}